\documentclass[12pt]{iopart}
\usepackage[utf8]{inputenc}
\usepackage{amsmath}
\usepackage{amsfonts}
\usepackage{amssymb}
\usepackage{amsthm}
\usepackage{graphicx}
\usepackage[table,dvipsnames]{xcolor}
\usepackage{url}
\usepackage{braket}
\usepackage{makecell,booktabs}
\usepackage{tikz}
\usepackage{hyperref}
\usepackage{wrapfig}
\usepackage{float}
\usepackage{array,longtable}
\usepackage[export]{adjustbox}
\usepackage{menukeys}
\usepackage[most]{tcolorbox}
\usepackage{tabularray}
\usepackage{caption}
\usepackage{subcaption}

\DeclareMathOperator{\ame}{AME}
\DeclareMathOperator{\iroa}{IrOA}

\DeclareMathOperator{\lu}{\stackrel{\text{LU}}{\sim}}

\DeclareMathOperator{\wfncomp}{\Psi}
\DeclareMathOperator{\wfnind}{\Psi}
\DeclareMathOperator{\wfncompp}{\Phi}
\DeclareMathOperator{\wfnindp}{\Phi}

\newcommand{\rl}[1]{\left(#1\right)}

\newcolumntype{?}[1]{!{\vrule width 2pt }}
\allowdisplaybreaks

\newtheorem{theorem}{Theorem}
\newtheorem{corollary}{Corollary}
\newtheorem{definition}{Definition}
\newtheorem{proposition}{Proposition}

\begin{document}

\title[]{Local unitary equivalence of absolutely maximally entangled states constructed from orthogonal arrays}
\author{N Ramadas${}^{1,2}$, Arul Lakshminarayan${}^{1,2}$}
\address{$^1$ Department of Physics, Indian Institute of Technology
Madras,\\ Chennai~600036, India}
\address{$^2$Center for Quantum Information, Communication and  Computation,\\ Indian Institute of Technology Madras, Chennai ~600036, India}
\eads{\mailto{ic37306@imail.iitm.ac.in}, \mailto{arul@physics.iitm.ac.in}}
\begin{abstract}
The classification of multipartite entanglement is essential as it serves as a resource for various quantum information processing tasks. This study concerns a particular class of highly entangled multipartite states, the so-called absolutely maximally entangled (AME) states. 
These are characterized by maximal entanglement across all possible bipartitions. In particular we analyze the local unitary equivalence among AME states using invariants. One of our main findings is that the existence of special  irredundant orthogonal arrays implies the existence of an infinite number of equivalence classes of AME states constructed from these. In particular, we show that there are infinitely many local unitary inequivalent three-party AME states for local dimension $d > 2$ and five-party AME states for $d \geq 2$.

\end{abstract}

\maketitle
\section{Introduction}
Practical applications of quantum theory, such as in quantum computing,  information processing and communication, typically involve many-body systems. Understanding the non-classical correlations in these systems is crucial, with quantum entanglement being particularly significant among them. While entanglement in pure bipartite states is well understood, it remains a challenge in the multipartite case. For example the notion of a maximally entangled state does not readily extend
as many particles can share entanglement differently due to the monogamy of entanglement. The 3 qubit GHZ state has vanishing entanglement among the qubits, but is considered highly entangled in the tripartite sense.

The present study focuses on one promising approach to the class of highly entangled pure multipartite states, called absolutely maximally entangled (AME) states \cite{helwigAbsoluteMaximalEntanglement2012}. AME states are such that there is maximal entanglement across any bipartition. There are an exponential (in number of particles) number of such bipartitions and these states are naturally very special in maximizing the entanglement in all these cases. They are useful in several applications including quantum error correction \cite{scottMultipartiteEntanglementQuantumerrorcorrecting2004}, quantum parallel transportation, secret sharing \cite{helwigAbsoluteMaximalEntanglement2012}, and holography \cite{pastawskiHolographicQuantumErrorcorrecting2015}. 
Thus special techniques have been devised to construct them, when possible. For example, they have been constructed from classical combinatorial designs such as orthogonal Latin squares and orthogonal arrays \cite{goyenecheGenuinelyMultipartiteEntangled2014, goyenecheAbsolutelyMaximallyEntangled2015,clarisseEntanglingPowerPermutations2005}.
A simple nontrivial example is the AME$(4,3)$ state, of four qutrits:
\begin{equation}
\label{eq:AME43}
\begin{split}
\ket{\Psi_{4,3}}=& \frac{1}{3}\left(\ket{0000}+\ket{0111}+\ket{0222}+\ket{1012}\right. +\\ &\left. \ket{1120}+\ket{1201}+\ket{2021}+\ket{2102}+\ket{2210}\right),
\end{split}
\end{equation}
which follows from an orthogonal Latin square of size $3$.
This is also an example of a {\it minimal support} AME state as it contains $d^{N/2}$ terms, every AME state will have at least these many terms in any single particle basis.

The determination of the existence of AME states of $N$ qudits, denoted as $\ame(N,d)$, remains open in general. This is number 35 in a list of open problems in quantum information \cite{OpenQuantumProblems}. A table of known constructions is provided in \cite{huberTableAMEStates}. In the case of qubits ($d=2$), it is established that AME states exist only for $N=2,3,5$, and $6$ \cite{higuchiHowEntangledCan2000, bennettMixedstateEntanglementQuantum1996,scottMultipartiteEntanglementQuantumerrorcorrecting2004,huberAbsolutelyMaximallyEntangled2017}. Only for these specific values of $N$, AME states exist for all local dimensions $d \geq 2$ \cite{goyenecheEntanglementQuantumCombinatorial2018}. In the case of $N=4$, AME states constructed using orthogonal Latin squares prove their existence for all $d\geq 3$, except for $d=6$. The problem of the existence of orthogonal Latin squares of order 6, famously known as ``Euler's 36 officers problem", has no classical solution \cite{tarryProbleme36Officiers1900}. A quantum solution, provided by an AME(4,6) state, was first given in \cite{ratherThirtysixEntangledOfficers2022}. The nature of the solution and geometric aspects were investigated in subsequent studies \cite{zyczkowskiUnderstandingQuantumSolution2023,zyczkowskiQuantumVersionEuler2023}. It was later demonstrated that an infinite number of inequivalent solutions exists \cite{ratherAbsolutelyMaximallyEntangled2023}. Some new solutions to the problem have recently been obtained \cite{ratherConstructionPerfectTensors2023,bruzdaTwounitaryComplexHadamard2024}.

The interconvertibility among states using local operations plays a crucial role in the understanding and classifying multipartite entanglement \cite{krausLocalUnitaryEquivalence2010}. Two $N$-party states $\ket{\Psi}$ and $\ket{\Psi'}$ are classified in the same local unitary (LU) class if they are intercovertible using LU operations, that is, there exist a set of local unitary operators $u_1,u_2,\ldots,u_N$ such that $ \ket{\Psi'} = \rl{u_1 \otimes u_2 \otimes \cdots \otimes u_N }\ket{\Psi} $. These local unitary operations do not change entanglement content in a state. When local operations, including measurements, are combined with classical communication among the parties, the entanglement is observed to remain non-increasing. This enables a classification using local operations and classical communication (LOCC) \cite{bennettConcentratingPartialEntanglement1996,bennettExactAsymptoticMeasures2000}, if a state $\ket{\psi}$ can be converted to $\ket{\phi}$ using LOCC, one may associate an ordering, with $\ket{\psi}$ having larger entanglement than $\ket{\phi}$. Two states are LOCC equivalent if they can be converted into each other. Stochastic LOCC (SLOCC)\cite{durThreeQubitsCan2000,acinGeneralizedSchmidtDecomposition2000} offers a coarser classification, examining whether a state can be converted to another using LOCC with a non-vanishing probability of success. Two states $\ket{\Psi}$ and $\ket{\Psi'}$ are SLOCC equivalent if $ \ket{\Psi'} = \rl{A_1 \otimes A_2 \otimes \cdots \otimes A_N }\ket{\Psi} $, where $A_1,A_2,\ldots,A_N$ are invertible matrices. 
A now classic result shows that 3 qubits can be in one of 6 SLOCC classes, of which 2 are genuinely entangled, the GHZ and W states \cite{durThreeQubitsCan2000}.
It is also known that 4 qubits have an infinity of SLOCC classes, sometimes classified into 9 families \cite{verstraeteFourQubitsCan2002}.

It has been noted that there is a connection between the geometry of an SLOCC class and the existence of a special type of state called a critical state within the class. A state is critical if all of its single-party reduced density matrices are proportional to the identity \cite{gourNecessarySufficientConditions2011},
otherwise called 1-uniform. The class of critical states includes stabilizer states, cluster states, and all $k$-uniform states, including AME states \cite{burchardtStochasticLocalOperations2020}.  
A consequence of the Kempf-Ness theorem \cite{kempfLengthVectorsRepresentation1979} is that a SLOCC class is topologically closed if and only if it contains a critical state, and that this is unique up to LU equivalence \cite{gourNecessarySufficientConditions2011}. Therefore, two critical states belong to the same SLOCC equivalence class if and only if they are LU equivalent. This implies that SLOCC (and thus LOCC) equivalence is equivalent to LU equivalence for AME states, as these states are special classes of critical states \cite{burchardtStochasticLocalOperations2020}. Consequently, focusing on LU equivalence is sufficient to classify the SLOCC equivalence classes of AME states.

The problem of LU equivalence of two given multi-qubit pure states was addressed by Kraus \cite{krausLocalUnitaryEquivalence2010} and later generalized for higher dimensions \cite{liuLocalUnitaryClassification2012}. However, these methods fail in the case of AME states due to the maximal entanglement present across any bipartition \cite{burchardtStochasticLocalOperations2020,krausLocalUnitaryEquivalence2010}. In the case of AME states, polynomial invariants \cite{grasslComputingLocalInvariants1998,rainsPolynomialInvariantsQuantum2000}, which remain unchanged under local unitary transformations, have been found useful.
A complete set of such polynomial LU invariants, proposed in \cite{kodiyalamCompleteSetNumerical2004}, was used in a previous study \cite{ratherAbsolutelyMaximallyEntangled2023} to investigate LU equivalence of four-party AME states. Using these invariants, it was demonstrated that there are infinitely many LU equivalence classes of AME$(4,d)$ states for all $d > 3$.

In this study, we use the complete set of LU invariants to further explore the LU equivalence of AME states. It is observed that these LU invariants in the complete set factorize in the case of composite states formed by combining qudits of lower dimensions. This property enables the study of their LU equivalence, particularly for composite AME states , which can be constructed through a similar procedure.
In \cite{ratherAbsolutelyMaximallyEntangled2023}, it was shown that the existence of a pair of orthogonal Latin squares of order $d$, satisfying the condition $d^2>4d-3$ implies the existence of an infinite number of $\ame(4,d)$ states. We extend this result to orthogonal arrays, which generalize  orthogonal Latin squares. In particular, AME states constructed from special irredundant orthogonal arrays \cite{goyenecheGenuinelyMultipartiteEntangled2014} are considered and found to imply the existence of infinitely many LU-inequivalent $\ame(N,d)$ states. 
This is applied, along with direct computation of LU invariants, to investigate the LU equivalence among $\ame(3,d)$ and $\ame(5,d)$ states. It is shown that there are infinitely many LU equivalence classes of $\ame(3,d)$ states for $d > 2$ and of $\ame(5,d)$ states for $d \geq 2$.

\section{Preliminaries and definitions}\label{sec:preliminaries}
In this section we review necessary background on AME states, irredundant orthogonal arrays, and local unitary equivalence.

\subsection{AME states}
Consider a system of $N$ particles with a total Hilbert space $\mathcal{H}^N_d =    \mathcal{H}^{(1)} \otimes \cdots \otimes \mathcal{H}^{(N)}$. The dimension of each Hilbert space is equal to $d$. A state of the system $\ket{\Psi}$ can be decomposed in the computational basis as
\begin{align}
\begin{aligned}
\ket{\Psi} = \sum_{j_1,j_2,\ldots,j_N =0}^{d-1} \Psi_{j_1,j_2,\ldots,j_N} \ket{j_1,j_2,\ldots,j_N},
\end{aligned}
\end{align}
where $\Psi_{j_1,j_2,\ldots,j_N} \in \mathbb{C}$. 

The aim is to understand the entanglement between various bipartitions of the system by computing the reduced density matrices and corresponding entanglement entropy. The reduced density matrix of a subsystem consisting of first $k$ parties is given by $\rho_{1,2,\ldots,k} = A A^\dagger$, where $A$ is a $d^k \times d^{N-k}$ matrix with elements
\begin{align}\label{eq:matricization}
\begin{aligned}
(A)^{j_1,j_2,\ldots,j_k}_{j_{k+1},\ldots,j_N}&= \braket{j_1,j_2,\ldots,j_k|A| j_{k+1},\ldots,j_N}  = \Psi_{j_1,j_2,\ldots,j_N}.
\end{aligned}
\end{align}
The reduced density matrix $\rho_S$ of any other subsystem of $k$ parties $S \subset [N]$, (the notation $[N]$ is used to denote the set $\lbrace 1,2,\ldots,N \rbrace $) can be related to a reshaping of the matrix $A$, as explained in what follows. The set $S$ can be written as $S = \lbrace \sigma(1),\sigma(2),\ldots,\sigma(k) \rbrace$, where $\sigma $ is an $N$-object permutation. The permutation $\sigma$ is not unique, but it results in the same $\rho_S$ as any rearrangement of remaining parties, those are traced out, does not matter. The action of the permutation $\sigma$ on the $N$-party state $\ket{\Psi}$ is given by
\begin{align}
\begin{aligned}
P_{\sigma} \ket{\Psi} = \sum_{j_1,j_2,\ldots,j_N =0}^{d-1} (A^{R_\sigma})^{j_{\sigma(1)},j_{\sigma(2)},\ldots,j_{\sigma(k)}}_{j_{\sigma(k+1)},\ldots,j_{\sigma(N)}} \ket{j_{\sigma(1)},j_{\sigma(2)},\ldots,j_{\sigma{(N)  }}},
\end{aligned}
\end{align}
where $P_{\sigma}$ is the permutation operator mapping a basis state $\ket{j_1,j_2,\ldots,j_N} $ to $ \ket{j_{\sigma(1)},j_{\sigma(2)},\ldots,j_{\sigma{(N)  }}}$. The permutation shuffles the parties in such a way that the all the parties belonging to subsystem $S$ are now the first $k$ parties of $P_{\sigma} \ket{\Psi}$. Hence, the reduced matrix $\rho_S = A^{R_\sigma} (A^{R_\sigma})^\dagger $, where
\begin{align}
\begin{aligned}\label{eq:reshaping}
(A^{R_\sigma})^{j_{\sigma(1)},j_{\sigma(2)},\ldots,j_{\sigma(k)}}_{j_{\sigma(k+1)},\ldots,j_{\sigma(N)}} = (A)^{j_1,j_2,\ldots,j_k}_{j_{k+1},\ldots,j_N}.
\end{aligned}
\end{align}
The above matrix reshaping operations are generalization of realignment and partial transpose operations \cite{chenMatrixRealignmentMethod2003,peresSeparabilityCriterionDensity1996}.

Computing the entropy of $\rho_S$ allows us to quantify the entanglement between subsystem $S$ and its complement $S^c$. Maximal entanglement between such partitions characterizes highly entangled states known as $k$-uniform states, and in particular, absolutely maximally states.

\begin{definition}
An N-party state $\ket{ \Psi} \in \mathcal{H}^{N}_{d} $ is $k$-uniform if the reduced density matrix of any subset $S$ of particles is $\rho_S =
\frac{1}{d^{|S|}} \mathbb{I}_{d^{|S|}}$, where $S \subset [N] $ and $|S| \leq k \leq \lfloor N/2 \rfloor $.
If $k=\lfloor N/2 \rfloor$, the state is an absolutely maximally entangled state, denoted by $\ame(N,d)$.
\end{definition}

For AME states, although $\rho_S$ is proportional to identity for all $S$ with $|S| \leq \lfloor N/2 \rfloor$, in practice, only for ${}^N C_{\lfloor N/2 \rfloor}$ sets of $S$ with $|S| = \lfloor N/2 \rfloor $ need to be checked, as partial traces of identity matrices remain identity matrices. The fact that the reduced density matrices $\rho_S$ are proportional to identity for all $S$ with $|S| = \lfloor N/2 \rfloor$ also implies that the reshaped matrices $A^{R_\sigma}$ are proportional to isometries for all $\sigma \in S_N$. For even $N$, the matrix $A$ associated with an AME state is called a multiunitary matrix \cite{goyenecheAbsolutelyMaximallyEntangled2015}.
The support of a state $\ket{\Psi}$ is the number of non-zero coefficients when $\ket{\Psi}$ is decomposed in the computational basis. 
 For an AME state, the support is at least $d^{ \lfloor N/2 \rfloor }$, which is attained for  \textit{minimal support} AME states
\cite{goyenecheAbsolutelyMaximallyEntangled2015,bernalserranoExistenceAbsolutelyMaximally2017}. 
An example of a minimal support AME state has already been presented in Eq.~(\ref{eq:AME43}), and these are significant due to their connection with Maximum Distance Separable (MDS) codes \cite{goyenecheAbsolutelyMaximallyEntangled2015}.
\subsection{Irredundant Orthogonal Arrays}\label{sec:iroa_definition}
A combinatorial design is an arrangement of elements from a finite set that satisfy certain properties \cite{stinsonCombinatorialDesignsConstructions2004}.  These designs have wide range of applications in many fields including optimal design of experiments \cite{raghavaraoConstructionsCombinatorialProblems1988,atkinsonOptimumExperimentalDesigns2007} and error correcting codes \cite{hedayatOrthogonalArraysTheory2012}. These are also useful in constructing $k$-uniform states and AME states \cite{goyenecheGenuinelyMultipartiteEntangled2014}.

Latin squares, orthogonal Latin squares, and orthogonal arrays \cite{colbournHandbookCombinatorialDesigns2007} are some examples of combinatorial designs. A Latin square is a $d \times d$ square arrangement of symbols $\lbrace0,1,\ldots,d-1 \rbrace$ such that each symbol occurs exactly once in each row and column. Two Latin squares are orthogonal if all pairs occur exactly once when they are superimposed. A set of pairwise orthogonal Latin squares form a set of mutually orthogonal Latin squares (MOLS). These constructions can be further generalized to mutually orthogonal Latin hypercubes and orthogonal arrays, first introduced by C.~R.~Rao in 1946 \cite{raoHypercubesStrengthLeading1946}.

\begin{definition}
An array with $r$ rows and $N$ columns with entries taken from a set of symbols $\lbrace 0,1,\ldots, d-1 \rbrace$ is an orthogonal array $OA(r,N,d,k)$ if each set of $k$ columns contains all $d^k$ possible combination of the symbols, each occurring $\lambda $ number of times along the rows \cite{hedayatOrthogonalArraysTheory2012}. The number $\lambda$ is called the index of the orthogonal array. It satisfies $\lambda= r/d^k$.
\end{definition}

An example of an orthogonal array is given by
\begin{align}\label{eq:example_oa}
\begin{aligned}
\text{OA}(9,3,3,2) = \begin{array}{ccc}
0&0&0\\ 0&1&1 \\0&2&2\\1&0&1 \\1&1&2\\1&2&0\\ 2&0&2\\2&1&0\\2&2&1
\end{array}.
\end{aligned}
\end{align}
In this array with 9 rows and 3 columns, the symbols are from the set $\lbrace 0,1,2 \rbrace$. Any two columns of the array contains all 9 pairs of these symbols, and the index of the orthogonal array $\lambda =1$. Note that this array is also $\text{OA}(9,3,3,1)$, with $\lambda=3$. In general an $\text{OA}(r,N,d,k)$ is also $\text{OA}(r,N,d,k')$ if $k'<k$.

We can construct the following three qutrit (spin-1) state from the orthogonal array:
\begin{align}
\label{eq:AME33}
\begin{aligned}
\ket{\Phi} &= \frac{1}{3} \bigg( \ket{0,0,0} +\ket{0,1,1} + \ket{0,2,2} + \ket{1,0,1} + \ket{1,1,2} \\&+ \ket{1,2,0}  + \ket{2,0,2} + \ket{2,1,0} + \ket{2,2,1}\bigg). 
\end{aligned}
\end{align}
It is straightforward to verify that the state is an $\ame(3,3)$ state by computing the reduced density matrices single parties. It also illustrates the construction of an AME state from an orthogonal array. However, not all orthogonal arrays can be used to construct AME states; they must satisfy additional constraints. This leads to the concept of irredundant orthogonal arrays.

\begin{definition}
An orthogonal array  $\text{OA}(r,N,d,k)$, is irredundant if every subset of $N-k$ columns contains no repeated rows \cite{goyenecheGenuinelyMultipartiteEntangled2014}. Such an array is specially  denoted as $\iroa(r,N,d,k)$.
\end{definition}
The example given above in Eq~\ref{eq:example_oa} is an irredundant orthogonal array given by $\iroa(9,3,3,1)$, meaning that no rows are repeating in any two columns of the array. Note that it is not an $\iroa(9,3,3,2)$, as the same array as  $\text{OA}(9,3,3,2)$ is not irredundant.  A given $\iroa(r,N,d,k)$ enables the construction of a $k$-uniform state. If the array is given by
\begin{align}
\begin{aligned}
\iroa(r,N,d,k) = 
\begin{array}{cccc}
s_{1}^{1}& s_{1}^{2}&\cdots& s_{1}^{N}\\
s_{2}^{1}& s_{2}^{2}&\cdots& s_{2}^{N}\\
\vdots& \vdots&\vdots &\vdots \\
s_{r}^{1}& s_{r}^{2}&\cdots& s_{r}^{N}\\
\end{array}
\end{aligned}
\end{align}
with $s^{i}_{j} \in \lbrace 0,1,\ldots,d-1\rbrace$ for $i=1,\ldots,r$ and $j=1,\ldots,N$,
then the state
\begin{align}
\begin{aligned}\label{eq:iroa}
\ket{\Psi} =\frac{1}{\sqrt{r}} \sum_{j = 1}^{r}  e^{i \theta_{\pmb{s}_j}} \ket{\pmb{s}_j},
\end{aligned}
\end{align}
where $ \pmb{s}_j = (s_{j}^{1} ,s_{j}^{2}, \cdots,s_{j}^{N})$ and $\theta_{\pmb{s}_j} \in \mathbb{R}$, is a $k$-uniform state. This can be seen as follows: since there are no repeated rows in any $N-k$ columns, tracing out $N-k$ parties leaves a sum of one-dimensional projections of states composed of rows of remaining $k$ columns. Since any $k$ columns contain all possible combinations of symbols repeating $\lambda$ times, the one-dimensional projections add up to an operator that is proportional to identity.

The classical combinatorial designs can be generalized to quantum combinatorial designs \cite{goyenecheEntanglementQuantumCombinatorial2018}. These are arrangements of quantum states that satisfy certain properties. Quantum Latin squares \cite{mustoQuantumLatinSquares2016}, quantum orthogonal Latin squares, and quantum orthogonal arrays\cite{goyenecheEntanglementQuantumCombinatorial2018} are some of the quantum combinatorial designs.

\subsection{LU equivalence and invariants of LU transformations}
An $N$-party state $\ket{\Psi}$ is LU equivalent to $\ket{\Psi'}$, denoted by $\ket{\Psi} \lu \ket{\Psi'} $, if and only if there exist local unitaries $u_i \in \mathbb{U}(d)$ for $i=1,2,\ldots,N$ such that
\begin{align}
\begin{aligned}
\ket{\Psi'} = (u_1 \otimes u_2 \otimes \cdots \otimes u_N) \ket{\Psi}.
\end{aligned}
\end{align}
Equivalently, the corresponding matrices $A$ and $A'$, as defined in Eq.~\ref{eq:matricization}, are LU equivalent if 
\begin{align}
\begin{aligned}
A' = (v_1 \otimes v_2 \otimes \cdots \otimes v_k) A (v_{k+1} \otimes \cdots \otimes v_N),
\end{aligned}
\end{align}
where $v_i \in \mathbb{U}(d)$ for $i=1,2,\ldots,N$.

Invariants of local unitary transformations are crucial for understanding LU equivalence among states. A function $f$ of states, $f: \mathcal{H}^{N}_{d} \to \mathbb{C}$ is a LU invariant if 
\begin{align}
\begin{aligned}
f( \ket{\Psi}) = f((u_1\otimes u_2 \otimes\cdots \otimes u_N) \ket{\Psi})
\end{aligned}
\end{align}
for all local unitary operators $u_1,u_2,\ldots,u_N \in \mathbb{U}(d)$. Consequently, if $f( \ket{\Psi}) \neq f( \ket{\Psi'})$, then the states $\ket{\Psi}$ and $\ket{\Psi'}$ are LU inequivalent. However, it is important to note that if the values of $f$ for two states are equal, it does not mean that they are LU equivalent. For this purpose, a complete set of LU invariants is required. Two states are LU equivalent if and only if all the LU invariants within the complete set are identical.

\section{A Complete set of LU invariants}\label{cslu}
A complete set of LU invariants allows, in principle, to determine whether two multipartite states are LU equivalent. Such a set that was introduced in \cite{kodiyalamCompleteSetNumerical2004} was used to study the LU equivalence of four-party AME states \cite{ratherAbsolutelyMaximallyEntangled2023}. General definitions of these invariants for multipartite state with any number of parties and dimension is provided in this section. Additionally, we derive a result concerning the LU invariants for composite AME states.

Let $\ket{\Psi} \in \mathcal{H}^{N}_{d}$ be a state of an $N$-party system, with each party described by a Hilbert space of dimension $d$. The LU invariant is defined by 
\begin{align}
\begin{aligned}\label{eq:luinvariant}
\Psi\rl{ \sigma_1,\sigma_2,\ldots,\sigma_N }& =  \prod_{l=1}^{n} \Psi_{ i^{(1)}_{l},i^{(2)}_{l},\ldots, i^{(N)}_{l}   } \Psi^*_{i^{(1)}_{\sigma_1(l)},i^{(2)}_{\sigma_2(l)},\ldots, i^{(N)}_{\sigma_N(l)} },
\end{aligned}
\end{align}
where $n \in \mathbb{N}$ is a natural number and $\sigma_1,\sigma_2,\ldots,\sigma_N \in S_n$ are permutation of $n$ objects. In the indices $i^{(j)}_{l} \in \{1,\ldots,d\}$, the superscript $j$ denotes the $j$-th party and the subscript $l$ denotes $l$-th copy of the system. The summation is performed over all the repeated indices in the above expression. Equivalently, we can provide a definition in term of matrices. If the $A$ is a matrix defined as in Eq.~\ref{eq:matricization}, the LU invariant is given by
\begin{align}
\begin{aligned}
A\rl { \sigma_1,\sigma_2,\ldots,\sigma_N }=  \prod_{l=1}^{n}{A}^{i^{(1)}_{l},\ldots, i^{(k)}_{l}   }_{i^{(k+1)}_{l},\ldots, i^{(N)}_{l}   } \rl{ A^{i^{(1)}_{ \sigma_1(l)},\ldots, i^{(k)}_{\sigma_k(l)} }_{i^{(k+1)}_{\sigma_{k+1}(l)},\ldots, i^{(N)}_{\sigma_N(l)} } }^* 
\end{aligned}
\end{align}
Here $k = \lfloor N/2 \rfloor$. From Eq.~\ref{eq:luinvariant}, it is clear that these are polynomial invariants \cite{grasslComputingLocalInvariants1998,rainsPolynomialInvariantsQuantum2000}.

A collection of these LU invariants for all natural numbers $n\in \mathcal{N}$ and all choices of permutations $\sigma_1,\sigma_2,\ldots,\sigma_N \in S_n$ contains a complete set of LU invariants \cite{kodiyalamCompleteSetNumerical2004}.
This implies two $N$-party states $\ket{\Psi}$ and  $\ket{\Psi'}$ are LU equivalent if and only if $ \Psi\rl{ \sigma_1,\sigma_2,\ldots,\sigma_N } = \Psi'\rl{ \sigma_1,\sigma_2,\ldots,\sigma_N }  $ for all $n\in \mathbb{N}$ and $\sigma_1,\sigma_2,\ldots,\sigma_N \in S_n$. Conversely, if they disagree on any LU invariant within the complete set, the states  $\ket{\Psi}$ and  $\ket{\Psi'}$ are LU inequivalent. 
\section{LU equivalence of composite AME states}
\label{sec:luinvofnewame}
Consider two systems, each comprising $N$ qudits, where the qudits in the first system are of dimension $d_1$ and those in the second system are of level $d_2$. Let these systems are represented by normalized states $ \ket{\wfnind^{(1)}} \in \mathcal{H}^N_{d_1}$ and $ \ket{\wfnind^{(2)}} \in \mathcal{H}^N_{d_2}$, respectively.
We can express these states as follows:
\begin{align}
\begin{aligned}
\ket{\wfnind^{(1)}} = \sum_{i_1,\ldots,i_{N} = 0}^{d_1-1}  \wfnind^{(1)}_{i_1,\ldots,i_{N}} \ket{i_1,\ldots,i_{N}}\;\text{and}\;
\ket{\wfnind^{(2)}} = \sum_{j_1,\ldots,j_{N} = 0}^{d_2-1}  \wfnind^{(2)}_{j_1,\ldots,j_{N}} \ket{j_1,\ldots,j_{N}},
\end{aligned}
\end{align}
where $\left\{ \ket{i_1,\ldots,i_{N}} \right\}$ and $\left\{ \ket{j_1,\ldots,j_{N}} \right\}$ are computational bases of $\mathcal{H}^N_{d_1}$ and $\mathcal{H}^N_{d_2}$, respectively.
Next, we construct a state $\ket{\Psi }\in \mathcal{H}^N_{d_1 d_2}$ of $N$ qudits of level $d_1 d_2$ by composing pairs of qudits from these two systems:
\begin{align}
\begin{aligned}
\ket{\wfncomp} &= \sum_{i_1,\ldots,i_{N} = 0}^{d_1-1}  \sum_{j_1,\ldots,j_{N} = 0}^{d_2-1} \wfnind^{(1)}_{i_1,\ldots,i_{N}} \wfnind^{(2)}_{j_1,\ldots,j_{N}}  \ket{i_1,j_1,i_2,j_2,\ldots,i_N,j_N}\\
& \equiv \sum_{l_1,\ldots,l_{N} = 0}^{d_1d_2-1}  \wfncomp_{l_1,\ldots,l_N}  \ket{l_1,\ldots,l_N},
\end{aligned}
\end{align}
where $l_k=(i_k,j_k)$ is a composite index which takes $d_1d_2$ values, and  $  \wfncomp_{l_1,\ldots,l_N} = \wfnind^{(1)}_{i_1,\ldots,i_{N}} \wfnind^{(2)}_{j_1,\ldots,j_{N}}$.
We denote this composition procedure as 
$\ket{\wfncomp} =P (\ket{\wfnind^{(1)}} \otimes \ket{\wfnind^{(2)}})$.

The density matrix $\rho = \ket{\wfncomp}\bra{\wfncomp}=$
\begin{align}
\begin{aligned}
\sum_{\substack{i_1,\ldots,i_N,\\ i'_1,\ldots,i'_N = 0}}^{d_1 -1} \sum_{\substack{j_1,\ldots,j_N,\\ j'_1,\ldots,j'_N = 0}}^{d_2 -1} \rho^{(1)}_{i_1,\ldots,i_N;i'_1,\ldots,i'_N}  \rho^{(2)}_{j_1,\ldots,j_N;j'_1,\ldots,j'_N} \ket{i_1,j_1,\ldots,i_N,j_N} \bra{ i'_1,j'_1,\ldots,i'_N,j'_N },
\end{aligned}
\end{align}
where $\rho^{(1)} = \ket{\wfnind^{(1)}}\bra{\wfnind^{(1)}}$ and $\rho^{(2)} = \ket{\wfnind^{(2)}}\bra{\wfnind^{(2)}}$ are the density matrices of the constituent systems.
From $\rho$ we can obtain the reduced density matrix $\rho_S$ of any subsystem $S$ of $k$ qudits:
\begin{align}
\begin{aligned}
\rho_S& = \sum_{\substack{i_1,\ldots,i_k,\\ i'_1,\ldots,i'_k = 0}}^{d_1 -1} \sum_{\substack{j_1,\ldots,j_k,\\ j'_1,\ldots,j'_k = 0}}^{d_2 -1} (\rho^{(1)}_{S_1})_{i_1,\ldots,i_k;i'_1,\ldots,i'_k} (\rho^{(2)}_{S_1})_{j_1,\ldots,j_k;j'_1,\ldots,j'_k} \ket{i_1,j_1,\ldots,i_k,j_k} \bra{i'_1,j'_1,\ldots,i'_k,j'_k},
\end{aligned}
\end{align}
where $S_1$ and $S_2$ are subsystems of $k$ qudits of level $d_1$ and $d_2$, respectively. The matrices $ \rho^{(1)}_{S_1} $ and $ \rho^{(2)}_{S_2} $ are corresponding reduced density matrices for these subsystems.

From the above relation, if $ \rho^{(1)}_{S_1} = \frac{1}{d_1^k} \mathbb{I}_{d_1^k} $ and $ \rho^{(2)}_{S_2} = \frac{1}{d_2^k} \mathbb{I}_{d_2^k} $, then $\rho_S = \frac{1}{(d_1d_2)^k} \mathbb{I}_{(d_1d_2)^k}$ for any subsystem $S$. Therefore, if $\ket{\wfnind^{(1)}}$ and $\ket{\wfnind^{(2)}}$ states are AME states, then $\ket{\wfncomp} =P (\ket{\wfnind^{(1)}} \otimes \ket{\wfnind^{(2)}})$ is an AME state. The construction of composite AME states from known AME states has been previously established, as noted in \cite{helwigAbsolutelyMaximallyEntangled2013}. The result can be generalized to construct composite states of more than two AME states.
\begin{proposition}
Let $ \ket{\wfncomp}  = P (\ket{\wfnind^{(1)}} \otimes \ket{\wfnind^{(2)}}) \in \mathcal{H}^{N}_{d_1 d_2}$ be an $N$-party state constructed using the above procedure from states $ \ket{\wfnind^{(1)}} \in \mathcal{H}^{N}_{d_1}$ and $ \ket{\wfnind^{(2)}} \in \mathcal{H}^{N}_{d_2}$. Then for any natural number $n$ and permutations $\sigma_1,\ldots,\sigma_N \in S_n$ the LU invariant in Eq.~\ref{eq:luinvariant} factors as
\begin{align}
\begin{aligned}
\wfncomp\rl{ \sigma_1,\sigma_2,\ldots,\sigma_N } =  \wfnind^{(1)}\rl{ \sigma_1,\sigma_2,\ldots,\sigma_N } \wfnind^{(2)}\rl{ \sigma_1,\sigma_2,\ldots,\sigma_N }.
\end{aligned}
\end{align}

\end{proposition}
\begin{proof}
From Eq.~\ref{eq:luinvariant} the LU invariant corresponding to the state $\ket{\wfncomp}$ is given by
\begin{align}
\begin{aligned}
\wfncomp\rl{ \sigma_1,\sigma_2,\ldots,\sigma_N }  &= \sum_{l^{(1)}_{1}, l^{(2)}_{1}, \cdots ,l^{(N)}_{n} = 0}^{d-1} \prod_{k=1}^{n} \wfncomp_{l^{(1)}_{k},l^{(2)}_{k},\cdots ,l^{(N)}_{k}} \wfncomp^*_{l^{(1)}_{\sigma_1(k)},l^{(2)}_{\sigma_2(k)},\ldots, l^{(N)}_{\sigma_N(k)}} \\
& =  \sum_{i^{(1)}_{1}, i^{(2)}_{1}, \cdots ,i^{(N)}_{n} = 0}^{d_1-1}
\sum_{j^{(1)}_{1}, j^{(2)}_{1}, \cdots ,j^{(N)}_{n}= 0}^{d_2-1} \prod_{k=1}^{n}   \wfnind^{(1)}_{ i^{(1)}_{k},\ldots,i^{(N)}_{k}} \wfnind^{(2)}_{ j^{(1)}_{k},\ldots,j^{(1)}_{k}} \\
&\quad \quad  \times (\wfnind^{(1)})^*_{ i^{(1)}_{\sigma_1(k) },\ldots,i^{(N)}_{\sigma_N(k)}} (\wfnind^{(2)})^*_{ j^{(1)}_{\sigma_1(k)},\ldots,j^{(N)}_{\sigma_N(k)}}  .
\end{aligned}
\end{align}
Simply rearranging gives
\begin{align}
\begin{aligned}
\wfncomp\rl{ \sigma_1,\sigma_2,\ldots,\sigma_N } =  \wfnind^{(1)}\rl{ \sigma_1,\sigma_2,\ldots,\sigma_N } \wfnind^{(2)}\rl{ \sigma_1,\sigma_2,\ldots,\sigma_N }
\end{aligned}
\end{align}
\end{proof}
The result is useful to surmise about the LU equivalence of composite states. Suppose $ \ket{\wfncomp}  =P (\ket{\wfnind^{(1)}} \otimes \ket{\wfnind^{(2)}}) \in \mathcal{H}_{d_1 d_2}^N$ and  $\ket{\wfncompp} =P \rl{\ket{\wfnindp^{(1)}} \otimes \ket{\wfnindp^{(2)}}} \in \mathcal{H}_{d_1 d_2}^N$, where $\ket{\wfnind^{(i)}} \in \mathcal{H}_{d_1}^N$ and $\ket{\wfnindp^{(i)}} \in \mathcal{H}_{d_2}^N$ for $i=1,2$. For any $n \in \mathbb{N}$ and choice of permutations $\sigma_1,\sigma_2,\ldots,\sigma_N \in S_n$, the following relations hold true:
\begin{align}\label{eq:two_composite_states}
\begin{aligned}
\wfncomp\rl{ \sigma_1,\sigma_2,\ldots,\sigma_N } &=  \wfnind^{(1)}\rl{ \sigma_1,\sigma_2,\ldots,\sigma_N } \wfnind^{(2)}\rl{ \sigma_1,\sigma_2,\ldots,\sigma_N },\\
\wfncompp\rl{ \sigma_1,\sigma_2,\ldots,\sigma_N } &=  \wfnindp^{(1)}\rl{ \sigma_1,\sigma_2,\ldots,\sigma_N } \wfnindp^{(2)}\rl{ \sigma_1,\sigma_2,\ldots,\sigma_N }.
\end{aligned}
\end{align}

If $\ket{\wfnind^{(i)}} \lu \ket{\wfnindp^{(i)}}$, then for all $n\in \mathbb{N}$ and for permutations $\sigma_1, \sigma_2, \ldots, \sigma_N \in S_n$, we have $\wfnind^{(i)}\rl{ \sigma_1,\sigma_2,\ldots,\sigma_N } = \wfnindp^{(i)}\rl{ \sigma_1,\sigma_2,\ldots,\sigma_N } $ for $i=1,2$. This implies $\wfncomp\rl{ \sigma_1, \sigma_2, \ldots, \sigma_N } = \wfncompp\rl{ \sigma_1, \sigma_2, \ldots, \sigma_N }$ for all such cases, and consequently, $\ket{\wfncomp} \lu \ket{\wfncompp}$. In the case where $d_1 = d_2$, if $\ket{\wfnind^{(1)}} \lu \ket{\wfnindp^{(2)}}$ and $\ket{\wfnind^{(2)}} \lu \ket{\wfnindp^{(1)}}$, a similar argument shows that $\ket{\wfncomp} \lu \ket{\wfncompp}$.

If $\ket{\Psi^{(1)}}$ is not LU equivalent to $\ket{\Phi^{(1)}}$, then there exists at least one $n\in N$ and $\sigma_1, \sigma_2, \ldots, \sigma_N \in S_n$ such that $\wfnind^{(1)}\rl{ \sigma_1, \sigma_2, \ldots, \sigma_N } \neq \wfnindp^{(1)}\rl{ \sigma_1, \sigma_2, \ldots, \sigma_N }$. For this choice of $n$ and permutations, two possibilities arise:
\begin{enumerate}
\item The invariants $\wfnind^{(2)}\rl{ \sigma_1, \sigma_2, \ldots, \sigma_N } $ and $ \wfnindp^{(2)}\rl{ \sigma_1, \sigma_2, \ldots, \sigma_N } $ are equal: In this case, if these invariants are non-zero, then $\wfncomp\rl{ \sigma_1, \sigma_2, \ldots, \sigma_N } \neq \wfncompp\rl{ \sigma_1, \sigma_2, \ldots, \sigma_N }$, implying $\ket{\wfncomp}$ is not LU equivalent to $\ket{\wfncompp}$. However, if they vanish, we cannot draw any conclusions about the LU equivalence of $\ket{\Phi}$ and $\ket{\Psi}$.
\item The invariants $\wfnind^{(2)}\rl{ \sigma_1, \sigma_2, \ldots, \sigma_N } \neq \wfnindp^{(2)}\rl{ \sigma_1, \sigma_2, \ldots, \sigma_N } $: In this case, it does not necessarily imply that $\wfncomp\rl{ \sigma_1, \sigma_2, \ldots, \sigma_N } \neq \wfncompp\rl{ \sigma_1, \sigma_2, \ldots, \sigma_N }$, since the product of the LU invariants can still be the same even if the individual LU invariants differ. 
\end{enumerate}
Note that if $\ket{\Psi^{(1)}}$ is not LU equivalent to $\ket{\Phi^{(1)}}$, there may be multiple possibilities for $n$ and permutations such that $\wfnind^{(1)}\rl{ \sigma_1, \sigma_2, \ldots, \sigma_N } \neq \wfnindp^{(1)}\rl{ \sigma_1, \sigma_2, \ldots, \sigma_N }$; each case must be examined to identify an $n$ and permutations $\sigma_1, \sigma_2, \ldots, \sigma_N \in S_n$ such that $\wfncomp\rl{ \sigma_1, \sigma_2, \ldots, \sigma_N } \neq \wfncompp\rl{ \sigma_1, \sigma_2, \ldots, \sigma_N }$. A similar analysis can be applied to other cases, such as when $\ket{\Psi^{(2)}}$ is not LU equivalent to $\ket{\Phi^{(2)}}$.

As mentioned previously, composite AME states can be constructed from known AME states using the procedure outlined above. The results presented here are applicable to these composite AME states, particularly in showing the existence of infinitely many LU equivalence classes of AME states. For example, consider a one-parameter family of AME$(N,d_1)$ states $\ket{\wfnind^{(1)}(\theta)}$ such that the LU invariant $\wfnind^{(1)}\rl{ \sigma_1,\sigma_2,\ldots,\sigma_N }$, for some choice of $n$ and permutations $\sigma_1,\sigma_2,\ldots,\sigma_N$, is a function of the parameter $\theta$. As distinct values of $\wfnind^{(1)}\rl{ \sigma_1,\sigma_2,\ldots,\sigma_N }$ correspond to different LU equivalence classes of $\ame(N,d_1)$ states, the number of these classes is infinite if the LU invariant takes infinitely many distinct values as $\theta$ varies. Since the square (or any power) of the LU invariant $\wfnind^{(1)}\rl{ \sigma_1,\sigma_2,\ldots,\sigma_N }$ remains a function of $\theta$, it implies the existence of an infinite number of LU equivalence class of $\ame(N,d_1^k)$ states for finite $k \geq 1$. Furthermore, if it is possible to find an $\ame(N,d_2)$ state $\ket{\wfnind^{(2)}}$ such that $ \wfnind^{(2)}\rl{ \sigma_1, \sigma_2, \ldots, \sigma_N } \neq 0 $ for the specified choice of $n$ and permutations, we can conclude that there are infinite number of LU equivalence classes of $\ame(N,d_1 d_2)$ states by the same argument.

\section{LU equivalence of AME states constructed from irredundant orthogonal arrays} \label{sec:iroa}
In an earlier work \cite{ratherAbsolutelyMaximallyEntangled2023}, the concept of multi-sets was used to demonstrate the existence of an infinite number of LU equivalence classes of AME$(4,d)$ states constructed from a pair of orthogonal Latin squares for all $d$, except $d=2$ and $d=6$. We extend this approach to orthogonal arrays, as they are generalizations of orthogonal Latin squares. In particular, we consider irredundant orthogonal arrays because they facilitate the construction of AME states.
\begin{theorem}\label{th:iroa}
If there exists an irredundant orthogonal array $IrOA(r,N,d,k)$ with $r > Nd - (N-1)$, then there are infinitely many LU equivalence classes of $k$-uniform states of $N$ parties of $d$ levels.
\end{theorem}
\begin{corollary}
If $k = \lfloor N/2 \rfloor, $ the number of LU equivalence classes of AME$(N,d)$ states is infinite.
\end{corollary}
\begin{proof}
In section \ref{sec:iroa_definition}, we discussed the construction of a $k$-uniform state from an irredundant orthogonal array. For completeness, we include the construction in this proof. Consider an irredundant orthogonal array IrOA$(r,N,d,k)$ given by the following $r \times N$ array:
\begin{align}
\begin{aligned}
\text{IrOA} (r,N,d,k) = \begin{array}{cccc}
s_{1}^{1}& s_{1}^{2}&\cdots& s_{1}^{N}\\
s_{2}^{1}& s_{2}^{2}&\cdots& s_{2}^{N}\\
\vdots& \vdots&\vdots &\vdots \\
s_{r}^{1}& s_{r}^{2}&\cdots& s_{r}^{N}.
\end{array}
\end{aligned}
\end{align}
Here, $s^{i}_{j} \in \lbrace 0,1,\ldots,d-1\rbrace$ for $i=1,\ldots,r$ and $j=1,\ldots,N$. Let $\mathcal{S} = \lbrace \pmb{s}_j | j=1,2,\ldots,r \rbrace$ denote the set of rows of the IrOA, 
where $\pmb{s}_j = (s_{j}^{1}, s_{j}^{1},\ldots,s_{j}^{N})$.
The IrOA allows the construction of a $k$-uniform state of $N$ parties each with $d$ levels, given by
\begin{align}
\begin{aligned}\label{eq:iroa}
\ket{\Psi} =\frac{1}{\sqrt{r}}\sum_{j = 1}^{r}  e^{i \theta_{\pmb{s}_j}} \ket{\pmb{s}_j},
\end{aligned}
\end{align}
where $\theta_{ \pmb{s}_j } \in \mathbb{R}$. In the special case when $k = \lfloor N/2 \rfloor$, it is an AME state.

For a given $n\in \mathbb{N}$ and permutations $\sigma_1,\sigma_2,\ldots,\sigma_N\in S_n$, the LU invariant, as defined in Eq~\ref{eq:luinvariant}, and repeated here for convenience, is given by 
\begin{align}\label{eq:luinvariant2}
\begin{aligned}
\Psi\rl{ \sigma_1,\sigma_2,\ldots,\sigma_N }& =  \prod_{l=1}^{n} \Psi_{ i^{(1)}_{l}i^{(2)}_{l}\cdots i^{(N)}_{l}   } \Psi^{*}_{i^{(1)}_{\sigma_1(l)}i^{(2)}_{\sigma_2(l)}\cdots i^{(N)}_{\sigma_N(l)} }.
\end{aligned}
\end{align}
The LU invariant $\Psi\rl{ \sigma_1,\sigma_2,\ldots,\sigma_N }$ is a polynomial in $e^{i \theta_{\pmb{s}_j}}$ and $e^{-i \theta_{\pmb{s}_j}}$ for all $\pmb{s}_j\in \mathcal{S}$. As distinct values of the LU invariant correspond to different LU equivalence classes, the number of these classes can be inferred from the range of the polynomial $\Psi\rl{ \sigma_1,\sigma_2,\ldots,\sigma_N }$. In particular, if the LU invariant is a non-constant polynomial in $e^{i \theta_{\pmb{s}_j}}$ and $e^{-i \theta_{\pmb{s}_j}}$ and takes infinitely many values as the parameters varied, this indicates the number of LU equivalence classes is infinite. However, it is not clear whether such an invariant exists in the infinite complete set of LU invariants. For example, all the elements in this complete set of LU invariants are constant functions corresponding to any parameter families of $\ame(4,3)$ states, as there is only one such state up to LU equivalence \cite{ratherAbsolutelyMaximallyEntangled2023}. (Parameter families of $\ame(4,3)$ states can be obtained from an $\iroa(9,4,3,2)$ using a similar construction given in Eq~\ref{eq:iroa}.)

A general term in the LU invariant $\Psi\rl{ \sigma_1,\sigma_2,\ldots,\sigma_N }$ in Eq~\ref{eq:luinvariant2} by suppressing the summation convention is
\begin{align*}
\begin{aligned}
\Xi_{X,Y} = \Psi_{ i^{(1)}_{1}i^{(2)}_{1}\cdots i^{(N)}_{1}} \times \cdots \Psi_{ i^{(1)}_{n}i^{(2)}_{n}\cdots i^{(N)}_{n}  } \times \Psi^{*}_{i^{(1)}_{\sigma_1(1)}i^{(2)}_{\sigma_2(1)}\cdots i^{(N)}_{\sigma_N(1)} } \times \cdots \times \Psi^{*}_{i^{(1)}_{\sigma_1(n)}i^{(2)}_{\sigma_2(n)}\cdots i^{(N)}_{\sigma_N(n)} },
\end{aligned}
\end{align*} 
where
$$X = \left\{ \rl{i^{(1)}_{l},i^{(2)}_{l},\ldots, i^{(N)}_{l}}  :l=1,\ldots,n \right\}$$ and $$Y = \left\{ \rl{i^{(1)}_{\sigma_1(l)}, i^{(2)}_{\sigma_2(l)}, \cdots,  i^{(N)}_{\sigma_N(l)} } :l=1,\ldots,n \right\}.$$ 
Note that $X$ and $Y$ are mult-sets, a generalization of sets, that allows for multiple occurrences of its elements. The term $\Xi_{X,Y}$ is non-vanishing if and only if
$X$ and $Y$ are multi-subsets of $\mathcal{S}$. Furthermore, if the sets $X$ and $Y$ are distinct, the phases $e^{i\theta_{\pmb{s}_j}}$ will not cancel out in the term. This can lead to the LU invariant in Eq.~\ref{eq:luinvariant2} being a non-constant polynomial in $e^{i \theta_{\pmb{s}_j}}$ and $e^{-i \theta_{\pmb{s}_j}}$. Note that these multi-subsets have same cardinality and are connected by the permutations $\sigma_1, \sigma_2,\ldots,\sigma_N$.

In what follows, we will identify necessary conditions for existence of such distinct multi-subsets $X$ and $Y$ and construct the corresponding AME state that yields a LU invariant with non-trivial dependence on the parameters. For convenience, we use the notation $\mathfrak{d} =\lbrace 0,1,\ldots,d-1\rbrace $. Define the functions $\pi_{\nu,l} : \mathfrak{d}^{\times N} \to \{0,1\}$ for $\nu = 1,2,\cdots N$ as follows:
\begin{align}
\begin{aligned}
\pi_{\nu,l} \rl{  s } = \begin{cases}
1 &~\text{if}~ s_\nu=l \\
0 &~\text{if}~ s_\nu \neq l
\end{cases},
\end{aligned}
\end{align}
where $s = (s_1,s_2,\ldots,s_N) \in \mathfrak{d}^{\times N} $. We use these functions is to count symbols at a given position. Our goal is to find two distinct multi-subsets $X$ and $Y$ of $\mathcal{S}$ of the same cardinality such that
\begin{align}
\begin{aligned}\label{eq:permutation_condition}
\sum_{s\in X}\pi_{\nu,l}(s) = \sum_{s\in Y} \pi_{\nu,l}(s),
\end{aligned}
\end{align} 
for all $\nu \in [N]$ and $l\in \mathfrak{d}$. These conditions ensure that the total count of a particular symbol $l$ occurring in the $\nu$-th position of tuples in both multi-subsets is the same, allowing the multi-subsets $X$ and $Y$ to be connected by permutations. 

Define the function $F: \mathfrak{d}^{\times N} \to \{0,1,2,3,\cdots\}$,
\begin{align}
\begin{aligned}
F(j_1,j_2,\ldots,j_N)  &= \sum_{ s \in X} \pi_{1,j_1}(s) \pi_{2,j_2}(s)\cdots\pi_{N,j_N}(s).
 \end{aligned}
\end{align}
This function counts the number of occurrences of a tuple $(j_1,j_2,\ldots,j_N)$ in the multi-subset $X$. Similarly, for $Y$, define $G: \mathfrak{d}^{\times N} \to \{0,1,2,3,\cdots\}$,
\begin{align}
\begin{aligned}
G(j_1,j_2,\ldots,j_N) & = \sum_{ s \in Y} \pi_{1,j_1}(s) \pi_{2,r_2}(s)\cdots\pi_{N,j_N}(s).
 \end{aligned}
\end{align}
Note that $\sum_{j_1,j_2,j_3,\ldots,j_N}F(j_1,j_2,\ldots,j_N) = |X|$ is the cardinality of X. Similarly, $\sum_{j_1,j_2,j_3,\ldots,j_N}G(j_1,j_2,\ldots,j_N) = |Y|$ is the cardinality of $Y$. Since $|X| = |Y|$, this implies
\begin{align}
\begin{aligned}\label{eq:total_number}
\sum_{j_1,j_2,j_3,\ldots,j_N}K(j_1,j_2,\ldots,j_N) = 0,
\end{aligned}
\end{align}
where $ K = F-G $.

We aim to express the equation Eq.~\ref{eq:permutation_condition} in terms of the function $K$. Observe that
\begin{align}
\begin{aligned}
\sum_{s\in X}\pi_{\nu,l}(s)  = \sum_{j_1,j_2,\ldots,j_N} F(j_1,j_2,\ldots,j_{\nu-1},l,j_{\nu+1},\ldots,j_N)
\end{aligned}
\end{align}
and
\begin{align}
\begin{aligned}
\sum_{s\in Y}\pi_{\nu,l}(s)  = \sum_{j_1,j_2,\ldots,j_N} G(j_1,j_2,\ldots,j_{\nu-1},l,j_{\nu+1},\ldots,j_N).
\end{aligned}
\end{align}
Taking the difference between the equations and expressing in terms of $ K(j_1,j_2,\ldots,j_N) $, we get
\begin{align}
\begin{aligned}
\sum_{j_1,j_2,\ldots,j_N} K(j_1,j_2,\ldots,j_{\nu-1},l,j_{\nu+1},\ldots,j_N) & = 0.
\end{aligned}
\end{align}
This gives a total of $Nd$ linear equations for various values of $\nu$ and $l$. However, $N-1$ of these equations are redundant, as they can be obtained from other equations. To see this, note that 
\begin{align}
\begin{aligned}
&\sum_{l}\sum_{j_1,j_2,\ldots,j_N } K(j_1,j_2,\ldots,j_{\nu-1},l,j_{\nu+1},\ldots,j_N) \\ & = \sum_{j_1,j_2,\ldots,j_N}K(j_1,j_2,\ldots,j_N) = 0.
\end{aligned}
\end{align}
These equations are same as the one in \ref{eq:total_number}. Therefore, the total number of equations is $Nd - (N-1)$. 
Solutions to these equations ensure existence of multi-subsets that satisfy the required conditions. To solve these equations, we can treat the terms $K(j_1,j_2,\ldots,j_N)$ as $r$ independent variables. A non-trivial solution to these set of equations exist if the number of variables is greater than the number of equations, that is, $r > Nd - (N-1)$. Since the coefficients of the equations are rational numbers, a rational solution exist. Appropriately rescaling the solutions gives an integral solution. The rescaled solutions are still a solution since the set of equations are homogeneous. Also, we choose $F(j_1,j_2,\ldots,j_N) $ and $G(j_1,j_2,\ldots,j_N)$ to be non-negative. Once the  multi-subsets $X$ and $Y$ are determined, the permutations $\sigma_1,\sigma_2,\ldots,\sigma_N$ can be constructed.

Assume the condition $r > Nd - (N-1)$ holds true. This implies there exist a tuple $(l_1,l_2,\ldots,l_N)$ such that $K(l_1,l_2,\ldots,l_N) \neq 0$. Therefore, if we choose 
\begin{align}
\begin{aligned}
\theta_{j_1,j_2,\ldots,j_N} = \begin{cases}
\theta &~~\text{if}~~ (j_1,j_2,\ldots,j_N) = (l_1,l_2,\ldots,l_N) \\
0 &~~\text{if}~~ (j_1,j_2,\ldots,j_N) \neq  (l_1,l_2,\ldots,l_N) ,
\end{cases}
\end{aligned}
\end{align}
there will be at least one term in the LU invariant that is a nonzero power of $e^{i\theta}$. This term will not cancel with other terms for general values of $\theta$ because all entries $\Psi_{j_1,j_2,\ldots,j_N}$ is positive except for $\Psi_{l_1,l_2,\ldots,l_N}$. Consequently, the LU invariant becomes a function of $\theta$ that takes infinitely many values as $\theta$ varies between $0$ and $2\pi$. Hence, there are an infinite number of LU equivalence classes.
\end{proof}
To illustrate the somewhat abstract arguments above, consider the AME$(3,3)$  state in Eq.~\ref{eq:AME33} constructed from the orthogonal array in Eq.~\ref{eq:example_oa}. The set of rows of this array, which is also IrOA(9,3,3,1), is
\begin{align*}
\mathcal{S} = \lbrace (0,0,0),(0,1,0),(0,2,2),(1,0,1),(1,1,2),(1,2,0),(2,0,2),(2,1,0),(2,2,1) \rbrace.
\end{align*}
The following are examples of ``multi-subsets" (in general there can be repetitions of set elements) used above:
\begin{align*}
X = \lbrace (0,0,0),(1,1,2),(2,2,1) \rbrace,\;
\text{and}\;
Y = \lbrace (0,2,2),(1,0,1),(2,1,0) \rbrace.
\end{align*}
The cardinality of these distinct sets are the same and they are connected by the permutations $\sigma_1 = (1~2~3),\sigma_2 = (3~1~2),$ and $\sigma_3 = (2~3~1)$ in the sense that if 
\begin{align*}
\begin{aligned}
X = \left\{ (i_1,j_1,k_1) ,(i_2,j_2,k_2),(i_3,j_3,k_3) \right\},
\end{aligned}
\end{align*} then \begin{align*}
\begin{aligned}
Y = \left\{ (i_{\sigma_1(1)},j_{\sigma_2(1)},k_{\sigma_3(1)}),(i_{\sigma_1(2)},j_{\sigma_2(2)},k_{\sigma_3(2)}),(i_{\sigma_1(3)},j_{\sigma_2(3)},k_{\sigma_3(3)}) \right\}.
\end{aligned}
\end{align*}
The above results shows that there are infinitely many LU equivalence classes of $\ame(3,3)$ states. The general case of $\ame(3,d)$ states is considered in the next section.

\begin{corollary}\label{cor:minimal}
If there exists an $\iroa(d^{\lfloor N/2 \rfloor}, N, d, \lfloor N/2 \rfloor)$—which is equivalent to having a minimal support $\ame(N, d)$ state—then the number of equivalence classes of $\ame(N, d)$ states is infinite for the following cases:\\
(1) $N = 4$ for all $d \geq 4$,\\
(2) $N = 5$ for all  $d\geq 5$,\\
(3) $N=7$ for all $d \geq 3$,\\
(4) $N=6$ and $N>8$ for all $d\geq 2$.
\end{corollary}
These cases can be identified by analyzing the condition $d^{\lfloor N/2 \rfloor}> Nd- (N-1)$ for various values of $N$ and $d$. For $d=2$, the condition is satisfied for $N=6$ and for all $N>8$. The other cases can be examined in a similar manner.
\section{LU equivalence classes of AME$(3,d)$ and AME$(5,d)$ states} \label{sec:results}
In this section, we apply the results discussed above to investigate the LU equivalence of AME states, focusing primarily on the cases where $N=3$ and $N=5$. It is known that among the SLOCC classes of three qubits states \cite{durThreeQubitsCan2000}, only the GHZ class contains AME states. Hence, all three qubit AME states are LU equivalent to the GHZ state:
\begin{align}
\begin{aligned}
\ket{GHZ} = \frac{1}{\sqrt{2}} \rl{ \ket{0,0,0} +\ket{1,1,1} }.
\end{aligned}
\end{align}
Thus, there is only one LU equivalence class of $\ame(3,2)$ states. The remaining cases are considered below.
\begin{theorem}
There exist an infinite number of LU equivalence classes of AME$(3,d)$ states for $d \geq 3$.
\end{theorem}
\begin{proof}
The state given by
\begin{align}
\begin{aligned}
\ket{\Psi_{3,d}} =\frac{1}{d} \sum_{j,k=0}^{d-1} \ket{j,k,j \oplus_d k},
\end{aligned}
\end{align}
where $\oplus_d $ denotes addition modulo $d$, is an AME$(3,d)$ state which can be constructed from an irredundant orthogonal array IrOA$(d^2,3,d,1)$ \cite{goyenecheEntanglementQuantumCombinatorial2018}. Since $d^2 > 3d-2$ for all $d\geq3$, by Theorem \ref{th:iroa}, there is an infinite number of LU equivalence classes of AME$(3,d)$ states for $d \geq 3$. 
\end{proof}

A direct computation of the LU invariant for a one-parameter family of $\ame(3,d)$ states confirms the same result. Consider the following AME$(3,d)$ state
\begin{align}
\begin{aligned}
\ket{\Psi_\theta} =\frac{1}{d} \sum_{j,k=0}^{d-1} e^{i \theta_{j,k} } \ket{ j,k,j \oplus_d k},
\end{aligned}
\end{align}
where $\theta_{0,0} = \theta$ and $ \theta_{j,k} =0 $ for all other values of $j$ and $k$. For $n = 3$ and the set of permutations
\begin{align}
\begin{aligned}
\sigma_1 = (1~ 2~ 3),~~
\sigma_2 = (2~ 3~ 1),~~
\sigma_3 &= (3~ 1~ 2),
\end{aligned}
\end{align}
the LU invariant in Eq.~\ref{eq:luinvariant} is computed as:
\begin{align}
\begin{aligned}
\Psi_{\theta}(\sigma_1,\sigma_2,\sigma_3) &= \frac{1}{d^6}\rl{d^4+ 6 (d-1)(d-2) ( \cos \theta -1)}.
\end{aligned}
\end{align}
The LU invariant is a non-trivial function of $\theta$ that takes infinitely many distinct values for $d>2$.

\begin{theorem}
There exist an infinite number of LU equivalence classes of AME$(5,d)$ states for $d \geq 2$.
\end{theorem}
\begin{proof}
It is known that IrOA($d^2,5,d,2$) exists for all $d\neq 2,3,6,10$ \cite{pangTwoThreeuniformStates2019}. Since the condition $d^2>5d-4$ from theorem \ref{th:iroa} holds true for IrOA($d^2,5,d,2$) for $d>4$, we conclude that an infinite number of LU equivalence classes of AME$(5,d)$ states exist for all $d$, except possibly, $d=2,3,4,6,10$. These specific cases require further investigation. Nonetheless, we directly compute the LU invariant for known AME$(5,d)$ states and prove the above result. The cases $d=2$ and $d>2$ are considered separately. 

For $d=2$, we consider the following AME$(5,2)$ state
\begin{align}
\begin{aligned}
\ket{\Psi_{5,2}} =\frac{1}{\sqrt{2}}\rl{ \cos\theta \ket{\tilde{0}} + \sin\theta \ket{\tilde{1}} }
\end{aligned}
\end{align}
where
\begin{align}
\begin{aligned}
\ket{\tilde{0}}  =& \frac{1}{\sqrt{8}} \left( \ket{0, 0, 0, 0, 0} + \ket{0, 0, 1, 1, 1} - \ket{0, 1, 0, 1, 0} + 
 \ket{0, 1, 1, 0, 1} \right.\\& \left. - \ket{1, 0, 0, 0, 1} - \ket{1, 0, 1, 1, 0} - 
 \ket{1, 1, 0, 1, 1} + \ket{1, 1, 1, 0, 0} \right)
\end{aligned}
\end{align}
and 
\begin{align}
\begin{aligned}
\ket{\tilde{1}}  =&\frac{1}{\sqrt{8}} \left( \ket{0, 0, 0, 1, 1} + \ket{0, 0, 1, 0, 0} - \ket{0, 1, 0, 0, 1} + 
 \ket{0, 1, 1, 1, 0}\right.\\& \left.+ \ket{1, 0, 0, 1, 0} + \ket{1, 0, 1, 0, 1}+ 
 \ket{1, 1, 0, 0, 0} - \ket{1, 1, 1, 1, 1} \right).
\end{aligned}
\end{align}
The $\ame(5,2)$ states $ \ket{\tilde{0}} $ and $ \ket{\tilde{1}} $ span a two-dimensional subspace of AME states, corresponding to a $[[5,1,3]]_2$ perfect error correcting code capable of correcting all Pauli errors \cite{laflammePerfectQuantumError1996}. We choose $n=5$ and the permutations $\sigma_1 = (1~2~3~4~5),
~\sigma_2 = (2~1~5~3~4),
~\sigma_3 = (3~4~1~5~2),
\sigma_4 = (4~5~2~1~3),$ and $
\sigma_5 = (5~3~4~2~1),$
to compute the LU invariant in Eq.~\ref{eq:luinvariant} as
\begin{align}
\begin{aligned}
\Psi_{5,2}\rl{ \sigma_1,\sigma_2,\ldots,\sigma_5} = -\frac{1}{8} (5 \cos (8 \theta)+3).
\end{aligned}
\end{align}
The LU invariant is a function of $\theta$ and takes infinitely many distinct values implying that there an infinite number of LU equivalence classes of AME$(5,2)$.

For $d>2$, we choose the following one-parameter family of AME$(5,d)$ states 
\begin{align}
\begin{aligned}
\ket{\Psi_{5,d}} =\frac{1}{d^{3/2}} \sum_{j,k=0}^{d-1} e^{\theta_{j,k}} \ket{ j,k,j \oplus_d k} \ket{\phi_{jk}}
\end{aligned}
\end{align}
where $\ket{\phi_{jk}} = \sum_{l=0}^{d-1} e^{2\pi i jl/d} \ket{l \oplus_d k,l} $, $ \theta_{0,0}=\theta $, and $ \theta_{j,k}=0 $ for all other values of $j$ and $k$.
The one-parameter family is constructed from a state given in \cite{goyenecheEntanglementQuantumCombinatorial2018}. By choosing $n =3$, and permutations $\sigma_1 = (1,2,3),
\sigma_2 = (2,3,1),
\sigma_3 = (3,1,2),
\sigma_4 = (2,3,1),
\sigma_5 = (1,2,3),$
LU invariant is
\begin{align}
\begin{aligned}
\Psi_{5,d}(\sigma_1,\sigma_2,\ldots,\sigma_5) & = \frac{1}{d^8}\rl{d^4+ 6 (d-1)(d-2) ( \cos \theta -1)}.
\end{aligned}
\end{align}
Again, as the LU invariant takes infinitely many distinct values for $d>2$, there are an infinite number of LU equivalence classes. 
\end{proof}

For $N=6$, minimal support states are known to exist for all $d\neq 2,3,6$ \cite{huberTableAMEStates}.  Hence by Corollary.~\ref{cor:minimal}, the number of LU equivalence classes of $\ame(6,d)$ states is infinite for these cases. A similar observation has been made in \cite{burchardtStochasticLocalOperations2020} using a different approach. It is known that, for qubits, there is only one $\ame(6,2)$ state up to LU equivalence \cite{rainsQuantumCodesMinimum1999}, leaving the cases $d=3$ and $d=6$ open.
%
\section{Conclusions}
The local unitary (LU) equivalence of absolutely maximally entangled (AME) states has been investigated using a complete set of LU invariants. It has been observed that the invariants in this complete set factorize for composite states constructed from lower-dimensional qudits. This factorization is useful to study the LU equivalence among composite AME states.

We have established the following result: if there exists an irredundant orthogonal array IrOA$(r,N,d,k)$ with $r>Nd- (N-1)$, then there are infinitely many LU equivalence classes of $k$-uniform states of $N$ parties, each with local dimension $d$. This result holds for AME states, which correspond to the case where $k = \lfloor N/2 \rfloor$. However, the existence and construction of irredundant orthogonal arrays are not known for many cases of $N$ and $d$. Several constructions can be found in \cite{goyenecheGenuinelyMultipartiteEntangled2014,goyenecheEntanglementQuantumCombinatorial2018,pangTwoThreeuniformStates2019,chenConstructionsIrredundantOrthogonal2023,chenNewResults2uniform2021} and related references.

Using the result, we showed that there are infinitely many LU inequivalent $\ame(3,d)$ states for all $d > 2$. In the case of $N=5$, irredundant orthogonal arrays that satisfy the criterion $r>5d- 4$ exist for all $d \neq 2,3,4,6,10$. This allows us to apply the result in these cases. Nonetheless, a direct computation of LU invariant for five-party AME states shows that there exist an infinite number of LU equivalence classes of $\ame(5,d)$ states for $d\geq 2$. With these new findings, an updated table of known minimal number of SLOCC inequivalent AME states is presented in Table.~\ref{tab:minimal_SLOCC}.
\begin{table}[h]
\begin{center}
\begin{tabular}{cc}
\begin{minipage}{0.6\textwidth}
\begin{tabular}{|c|c|c|c|c|c|c|}
\hline
\cellcolor{yellow!20}&\cellcolor{yellow!20}2&\cellcolor{yellow!20}3&\cellcolor{yellow!20}4&\cellcolor{yellow!20}5&\cellcolor{yellow!20}6&\cellcolor{yellow!20}7\\
\hline 
\cellcolor{yellow!20}AME(3,d)&\cellcolor{green!20}1&\cellcolor{blue!20}$\infty$&\cellcolor{blue!20}$\infty$&\cellcolor{blue!20}$\infty$&\cellcolor{blue!20}$\infty$&\cellcolor{blue!20}$\infty$\\
\hline 
\cellcolor{yellow!20}AME(4,d)&\cellcolor{green!20}0&\cellcolor{green!20}1& \cellcolor{green!20} $\infty$&\cellcolor{green!20} $\infty$& \cellcolor{green!20}$\infty$& \cellcolor{green!20}$\infty$\\
\hline 
\cellcolor{yellow!20}AME(5,d)&\cellcolor{blue!20} $\infty$&\cellcolor{blue!20}$\infty$&\cellcolor{blue!20}$\infty$&\cellcolor{blue!20}$\infty$&\cellcolor{blue!20}$\infty$&\cellcolor{blue!20}$\infty$\\
\hline 
\cellcolor{yellow!20}AME(6,d)&\cellcolor{green!20}1&1&\cellcolor{green!20}$\infty$&\cellcolor{green!20}$\infty$&1& \cellcolor{green!20} $\infty$\\
\hline 
\cellcolor{yellow!20}AME(7,d)&\cellcolor{green!20} 0&1&1&1& \cellcolor{red!20}?&\cellcolor{green!20}$\infty$\\
\hline
\end{tabular}
\end{minipage}
&
\begin{minipage}{0.4\textwidth}
\tikz[baseline=(current bounding box.north)]{
    \fill[blue!20] (0,0) rectangle (0.5,0.5); 
    \draw[black] (0,0) rectangle (0.5,0.5); 
    \node[anchor=west] at (0.75,0.25) {New updates}; 
} \\
\tikz[baseline=(current bounding box.north)]{
    \fill[green!20] (0,0) rectangle (0.5,0.5); 
    \draw[black] (0,0) rectangle (0.5,0.5); 
    \node[anchor=west] at (0.75,0.25) {Proven result}; 
}
\\
\tikz[baseline=(current bounding box.north)]{
    \fill[white] (0,0) rectangle (0.5,0.5); 
    \draw[black] (0,0) rectangle (0.5,0.5); 
    \node[anchor=west] at (0.75,0.25) {Minimum known value}; 
}
\\
\tikz[baseline=(current bounding box.north)]{
    \fill[red!20] (0,0) rectangle (0.5,0.5); 
    \draw[black] (0,0) rectangle (0.5,0.5); 
    \node[anchor=west] at (0.75,0.25) {Unknown}; 
	\node at (0.25,0.25) {?};
}
\end{minipage}
\end{tabular}
\end{center}

\caption{The minimal number of SLOCC-inequivalent AME states is provided. A number 0 denotes the non-existence of the relevant state, and a question mark signifies that the status is unknown.}
\label{tab:minimal_SLOCC}
\end{table}

%

The orthogonal arrays, considered in this work, are classical combinatorial designs. An interesting question is to generalize Theorem.~\ref{th:iroa} to quantum orthogonal arrays \cite{goyenecheEntanglementQuantumCombinatorial2018}. This will allow us to study equivalence of AME states for cases of $N$ and $d$ where constructions from orthogonal arrays may not exist.
\section{Acknowledgement}
N.R. acknowledges funding from the Center for Quantum Information, Communication, and Computing, IIT Madras. 
We thank Arun Joseph for discussions on the topic.

\section*{References}

\bibliography{global_references.bib}

\begin{thebibliography}{10}

\bibitem{helwigAbsoluteMaximalEntanglement2012}
Wolfram Helwig, Wei Cui, Jos{\'e}~Ignacio Latorre, Arnau Riera, and Hoi-Kwong
  Lo.
\newblock Absolute maximal entanglement and quantum secret sharing.
\newblock {\em Physical Review A}, 86(5):052335, November 2012.

\bibitem{scottMultipartiteEntanglementQuantumerrorcorrecting2004}
A.~J. Scott.
\newblock Multipartite entanglement, quantum-error-correcting codes, and
  entangling power of quantum evolutions.
\newblock {\em Physical Review A}, 69(5):052330, May 2004.

\bibitem{pastawskiHolographicQuantumErrorcorrecting2015}
Fernando Pastawski, Beni Yoshida, Daniel Harlow, and John Preskill.
\newblock Holographic quantum error-correcting codes: Toy models for the
  bulk/boundary correspondence.
\newblock {\em Journal of High Energy Physics}, 2015(6):149, June 2015.

\bibitem{goyenecheGenuinelyMultipartiteEntangled2014}
Dardo Goyeneche and Karol {\.Z}yczkowski.
\newblock Genuinely multipartite entangled states and orthogonal arrays.
\newblock {\em Physical Review A}, 90(2):022316, August 2014.

\bibitem{goyenecheAbsolutelyMaximallyEntangled2015}
Dardo Goyeneche, Daniel Alsina, Jos{\'e}~I. Latorre, Arnau Riera, and Karol
  {\.Z}yczkowski.
\newblock Absolutely maximally entangled states, combinatorial designs, and
  multiunitary matrices.
\newblock {\em Physical Review A}, 92(3):032316, September 2015.

\bibitem{clarisseEntanglingPowerPermutations2005}
Lieven Clarisse, Sibasish Ghosh, Simone Severini, and Anthony Sudbery.
\newblock Entangling power of permutations.
\newblock {\em Physical Review A}, 72(1):012314, July 2005.

\bibitem{OpenQuantumProblems}
Open {{Quantum Problems}} -- {{Open Quantum Problems}} :
  {{https://oqp.iqoqi.oeaw.ac.at/open-quantum-problems}}.

\bibitem{huberTableAMEStates}
F~Huber and N~Wyderka.
\newblock Table of {{AME}} states / perfect tensors.
\newblock https://tp.nt.uni-siegen.de/ame/ame.html.

\bibitem{higuchiHowEntangledCan2000}
A.~Higuchi and A.~Sudbery.
\newblock How entangled can two couples get?
\newblock {\em Physics Letters A}, 273(4):213--217, August 2000.

\bibitem{bennettMixedstateEntanglementQuantum1996}
Charles~H. Bennett, David~P. DiVincenzo, John~A. Smolin, and William~K.
  Wootters.
\newblock Mixed-state entanglement and quantum error correction.
\newblock {\em Physical Review A}, 54(5):3824--3851, November 1996.

\bibitem{huberAbsolutelyMaximallyEntangled2017}
Felix Huber, Otfried G{\"u}hne, and Jens Siewert.
\newblock Absolutely {{Maximally Entangled States}} of {{Seven Qubits Do Not
  Exist}}.
\newblock {\em Physical Review Letters}, 118(20):200502, May 2017.

\bibitem{goyenecheEntanglementQuantumCombinatorial2018}
Dardo Goyeneche, Zahra Raissi, Sara Di~Martino, and Karol Zyczkowski.
\newblock Entanglement and quantum combinatorial designs.
\newblock {\em Physical Review A}, 97(6):062326, June 2018.

\bibitem{tarryProbleme36Officiers1900}
Gaston Tarry.
\newblock Le probl{\`e}me des 36 officiers.
\newblock {\em Secr{\'e}tariat de l'Association fran{\c c}aise pour
  l'avancement des sciences}, 1900.

\bibitem{ratherThirtysixEntangledOfficers2022}
Suhail~Ahmad Rather, Adam Burchardt, Wojciech Bruzda, Grzegorz
  {Rajchel-Mieldzio{\'c}}, Arul Lakshminarayan, and Karol {\.Z}yczkowski.
\newblock Thirty-six {{Entangled Officers}} of {{Euler}}: {{Quantum Solution}}
  to a {{Classically Impossible Problem}}.
\newblock {\em Physical Review Letters}, 128(8):080507, February 2022.

\bibitem{zyczkowskiUnderstandingQuantumSolution2023}
K.~{\.Z}yczkowski, W.~Bruzda, G.~{Rajchel-Mieldzio{\'c}}, A.~Burchardt,
  S.~Ahmad Rather, and A.~Lakshminarayan.
\newblock 9 {\texttimes} 4 = 6 {\texttimes} 6: {{Understanding}} the {{Quantum
  Solution}} to {{Euler}}'s {{Problem}} of 36 {{Officers}}.
\newblock {\em Journal of Physics: Conference Series}, 2448(1):012003, February
  2023.

\bibitem{zyczkowskiQuantumVersionEuler2023}
Karol {\.Z}yczkowski.
\newblock Quantum {{Version}} of {{Euler}}'s {{Problem}}: {{A Geometric
  Perspective}}.
\newblock In Piotr Kielanowski, Alina Dobrogowska, Gerald~A. Goldin, and Tomasz
  Goli{\'n}ski, editors, {\em Geometric {{Methods}} in {{Physics XXXIX}}},
  pages 105--133, Cham, 2023. Springer International Publishing.

\bibitem{ratherAbsolutelyMaximallyEntangled2023}
Suhail~Ahmad Rather, N.~Ramadas, Vijay Kodiyalam, and Arul Lakshminarayan.
\newblock Absolutely maximally entangled state equivalence and the construction
  of infinite quantum solutions to the problem of 36 officers of {{Euler}}.
\newblock {\em Physical Review A}, 108(3):032412, September 2023.

\bibitem{ratherConstructionPerfectTensors2023}
Suhail~Ahmad Rather.
\newblock Construction of perfect tensors using biunimodular vectors, September
  2023.

\bibitem{bruzdaTwounitaryComplexHadamard2024}
Wojciech Bruzda and Karol {\.Z}yczkowski.
\newblock Two-unitary complex {{Hadamard}} matrices of order 36.
\newblock {\em Special Matrices}, 12(1), January 2024.

\bibitem{krausLocalUnitaryEquivalence2010}
B.~Kraus.
\newblock Local {{Unitary Equivalence}} of {{Multipartite Pure States}}.
\newblock {\em Physical Review Letters}, 104(2):020504, January 2010.

\bibitem{bennettConcentratingPartialEntanglement1996}
Charles~H. Bennett, Herbert~J. Bernstein, Sandu Popescu, and Benjamin
  Schumacher.
\newblock Concentrating partial entanglement by local operations.
\newblock {\em Physical Review A}, 53(4):2046--2052, April 1996.

\bibitem{bennettExactAsymptoticMeasures2000}
Charles~H. Bennett, Sandu Popescu, Daniel Rohrlich, John~A. Smolin, and
  Ashish~V. Thapliyal.
\newblock Exact and asymptotic measures of multipartite pure-state
  entanglement.
\newblock {\em Physical Review A}, 63(1):012307, December 2000.

\bibitem{durThreeQubitsCan2000}
W.~D{\"u}r, G.~Vidal, and J.~I. Cirac.
\newblock Three qubits can be entangled in two inequivalent ways.
\newblock {\em Physical Review A}, 62(6):062314, November 2000.

\bibitem{acinGeneralizedSchmidtDecomposition2000}
A.~Ac{\'i}n, A.~Andrianov, L.~Costa, E.~Jan{\'e}, J.~I. Latorre, and
  R.~Tarrach.
\newblock Generalized {{Schmidt Decomposition}} and {{Classification}} of
  {{Three-Quantum-Bit States}}.
\newblock {\em Physical Review Letters}, 85(7):1560--1563, August 2000.

\bibitem{verstraeteFourQubitsCan2002}
F.~Verstraete, J.~Dehaene, B.~De~Moor, and H.~Verschelde.
\newblock Four qubits can be entangled in nine different ways.
\newblock {\em Physical Review A}, 65(5):052112, April 2002.

\bibitem{gourNecessarySufficientConditions2011}
Gilad Gour and Nolan~R. Wallach.
\newblock Necessary and sufficient conditions for local manipulation of
  multipartite pure quantum states.
\newblock {\em New Journal of Physics}, 13(7):073013, July 2011.

\bibitem{burchardtStochasticLocalOperations2020}
Adam Burchardt and Zahra Raissi.
\newblock Stochastic local operations with classical communication of
  absolutely maximally entangled states.
\newblock {\em Physical Review A}, 102(2):022413, August 2020.

\bibitem{kempfLengthVectorsRepresentation1979}
George Kempf and Linda Ness.
\newblock The length of vectors in representation spaces.
\newblock In Knud L{\o}nsted, editor, {\em Algebraic {{Geometry}}}, pages
  233--243, Berlin, Heidelberg, 1979. Springer.

\bibitem{liuLocalUnitaryClassification2012}
Bin Liu, Jun-Li Li, Xikun Li, and Cong-Feng Qiao.
\newblock Local {{Unitary Classification}} of {{Arbitrary Dimensional
  Multipartite Pure States}}.
\newblock {\em Physical Review Letters}, 108(5):050501, January 2012.

\bibitem{grasslComputingLocalInvariants1998}
Markus Grassl, Martin Roetteler, and Thomas Beth.
\newblock Computing {{Local Invariants}} of {{Qubit Systems}}.
\newblock {\em Physical Review A}, 58(3):1833--1839, September 1998.

\bibitem{rainsPolynomialInvariantsQuantum2000}
E.M. Rains.
\newblock Polynomial invariants of quantum codes.
\newblock {\em IEEE Transactions on Information Theory}, 46(1):54--59, January
  2000.

\bibitem{kodiyalamCompleteSetNumerical2004}
Vijay Kodiyalam and V.~S. Sunder.
\newblock A complete set of numerical invariants for a subfactor.
\newblock {\em Journal of Functional Analysis}, 212(1):1--27, July 2004.

\bibitem{chenMatrixRealignmentMethod2003}
Kai Chen and Ling-An Wu.
\newblock A matrix realignment method for recognizing entanglement, April 2003.

\bibitem{peresSeparabilityCriterionDensity1996}
Asher Peres.
\newblock Separability {{Criterion}} for {{Density Matrices}}.
\newblock {\em Physical Review Letters}, 77(8):1413--1415, August 1996.

\bibitem{bernalserranoExistenceAbsolutelyMaximally2017}
Antonio Bernal~Serrano.
\newblock On the {{Existence}} of {{Absolutely Maximally Entangled States}} of
  {{Minimal Support}}.
\newblock {\em Articles publicats en revistes (Matem{\`a}tiques i
  Inform{\`a}tica)}, April 2017.

\bibitem{stinsonCombinatorialDesignsConstructions2004}
Douglas~R. Stinson.
\newblock {\em Combinatorial Designs: Constructions and Analysis}.
\newblock Springer, New York, 2004.

\bibitem{raghavaraoConstructionsCombinatorialProblems1988}
Damaraju Raghavarao.
\newblock {\em Constructions and Combinatorial Problems in Design of
  Experiments}.
\newblock Dover Publications, New York, 1988.

\bibitem{atkinsonOptimumExperimentalDesigns2007}
Anthony~Curtis Atkinson, Alexander~N. Donev, and Randall~Davis Tobias.
\newblock {\em Optimum Experimental Designs, with {{SAS}}}.
\newblock Number~34 in Oxford Statistical Science Series. Oxford university
  press, Oxford, 2007.

\bibitem{hedayatOrthogonalArraysTheory2012}
A.~S. Hedayat, N.~J.~A. Sloane, and John Stufken.
\newblock {\em Orthogonal {{Arrays}}: {{Theory}} and {{Applications}}}.
\newblock Springer Science \& Business Media, December 2012.

\bibitem{colbournHandbookCombinatorialDesigns2007}
C.~J. Colbourn and Jeffrey~H. Dinitz.
\newblock {\em Handbook of Combinatorial Designs}.
\newblock Chapman \& Hall/Taylor \& Francis, Boca Raton, FL, 2nd ed edition,
  2007.

\bibitem{raoHypercubesStrengthLeading1946}
C~Radhakrishna Rao.
\newblock Hypercubes of strength 'd' leading to confounded designs in factorial
  experiments.
\newblock {\em Bull. Calcutta Math. Soc.}, 38:67--78, 1946.

\bibitem{mustoQuantumLatinSquares2016}
Benjamin Musto and Jamie Vicary.
\newblock Quantum {{Latin}} squares and unitary error bases.
\newblock {\em Quantum Information and Computation}, 16(15\&16):1318--1332,
  November 2016.

\bibitem{helwigAbsolutelyMaximallyEntangled2013}
Wolfram Helwig.
\newblock Absolutely {{Maximally Entangled Qudit Graph States}}, June 2013.

\bibitem{pangTwoThreeuniformStates2019}
Shan-Qi Pang, Xiao Zhang, Xiao Lin, and Qing-Juan Zhang.
\newblock Two and three-uniform states from irredundant orthogonal arrays.
\newblock {\em npj Quantum Information}, 5(1):1--10, June 2019.

\bibitem{laflammePerfectQuantumError1996}
Raymond Laflamme, Cesar Miquel, Juan~Pablo Paz, and Wojciech~Hubert Zurek.
\newblock Perfect {{Quantum Error Correcting Code}}.
\newblock {\em Physical Review Letters}, 77(1):198--201, July 1996.

\bibitem{rainsQuantumCodesMinimum1999}
E.M. Rains.
\newblock Quantum codes of minimum distance two.
\newblock {\em IEEE Transactions on Information Theory}, 45(1):266--271,
  January 1999.

\bibitem{chenConstructionsIrredundantOrthogonal2023}
Guangzhou Chen and Xiaotong Zhang.
\newblock Constructions of irredundant orthogonal arrays.
\newblock {\em Advances in Mathematics of Communications}, 17(6):1314--1337,
  Fri Dec 01 05:00:00 UTC 2023.

\bibitem{chenNewResults2uniform2021}
Guangzhou Chen, Xiaotong Zhang, and Yue Guo.
\newblock New results for 2-uniform states based on irredundant orthogonal
  arrays.
\newblock {\em Quantum Information Processing}, 20(1):43, January 2021.

\end{thebibliography}
\bibliographystyle{unsrt}
\end{document}